\documentclass[letterpaper, 10 pt, conference]{ieeeconf}  

\IEEEoverridecommandlockouts                              

\usepackage{cite}
\usepackage{amsmath,amssymb,amsfonts}
\usepackage{algorithmic}
\usepackage{graphicx}
\usepackage{textcomp}
\usepackage{xcolor}
\usepackage{soul}

\renewcommand{\baselinestretch}{0.96}

\newtheorem{remark}{\textbf{Remark}}
\newtheorem{lem}{\textbf{Lemma}}
\newtheorem{theorem}{\textbf{Theorem}}
\newtheorem{proposition}{\textbf{Proposition}}
\newtheorem{assumption}{\textbf{Assumption}}
\newtheorem{definition}{\textbf{Definition}}

\title{\LARGE \bf
Mitigation and Resiliency  of Multi-Agent Systems Subject to Malicious Cyber Attacks on Communication Links}

\author{Mahdi Taheri$^{1}$, Khashayar Khorasani$^{1}$, Iman Shames$^{2}$, and Nader Meskin$^{3}$
\thanks{$^{1}$Mahdi Taheri (m$\_$eri@encs.concordia.ca) and Khashayar Khorasani (kash@ece.concordia.ca) are with the Department of Electrical and Computer Engineering, Concordia University, Montreal, Canada.}%
\thanks{$^{2}$Iman Shames (iman.shames@unimelb.edu.au) is with the Department of Electrical and Electronic Engineering, University of Melbourne, Melbourne, Australia.}%
\thanks{$^{3}$Nader Meskin (nader.meskin@qu.edu.qa) is with the Department of Electrical Engineering, Qatar University, Doha, Qatar.}%
\thanks{The authors would like to acknowledge the financial support received from NATO under the Emerging Security Challenges Division program. K. Khorasani and N. Meskin would like to acknowledge the support received from NPRP grant number 10-0105-17017 from the Qatar National Research Fund (a member of Qatar Foundation). K. Khorasani would also like to acknowledge the support received from the Natural Sciences and Engineering Research Council of Canada (NSERC) and the Department of National Defence (DnD) under the Discovery Grant and DnD Supplemental Programs. The statements made herein are solely the responsibility of the authors.}
}

\begin{document}

\maketitle
\thispagestyle{empty}
\pagestyle{empty}

\begin{abstract}
	This paper aims at investigating a novel type of cyber attack that is injected to multi-agent systems (MAS) having an underlying directed graph. The cyber attack, which is designated as the controllability attack, is injected by the malicious adversary  into the communication links among the agents. The adversary, leveraging the compromised communication links disguises the cyber attack signals and attempts to take control over the entire network of MAS. The adversary aims at achieving this by directly attacking only a subset of the multi-agents. Conditions under which the malicious hacker has control over the entire MAS network are provided. Two notions of security controllability indices are proposed and developed. These notions are utilized as metrics to evaluate the controllability that each agent provides to the adversary for executing the malicious cyber attack. Furthermore, the possibility of introducing zero dynamics cyber attacks on the MAS through compromising the communication links is also investigated. Finally, an illustrative numerical example is provided to demonstrate the effectiveness of our proposed methods.
\end{abstract}

\section{Introduction}
Multi-agent systems (MAS), due to their wide range of applications, such as in unnamed aerial vehicles (UAV), next generation aerospace and transportation systems, autonomous and drive-less cars, have been a major topic of research during the past decade \cite{1333204,4804653,DAVOODI2016185,dfdfi,semsar2009multi}. One of the challenges in MAS is to reach a consensus among the agents in a distributed manner. This problem has been addressed for systems having various types of linear and nonlinear dynamics \cite{1333204,xu2013consensus,taheri2017adaptive,li2009consensus}. To achieve consensus among agents, each agent needs to transmit its information to its nearest neighboring agents. This communication is carried out through network channels that exist among the agents.

Existence of  communication networks make the multi-agent systems to be vulnerable to cyber attacks. Suppose a group of agents are on an intelligence, surveillance, and reconnaissance (ISR) mission and an intelligent adversary performs an attack on the incoming communication links for a subset of these agents. The adversary, using the incoming communication signals can directly modify the received data associated with the compromised agents. In such a scenario, one is interested in characterizing conditions under which the adversary is capable of taking over and controlling the remainder of the agents. This is the main question that we are investigating and providing solutions to in this paper.

The topic of security in cyber-physical systems (CPS) has received a considerable amount of attention in recent years \cite{adaii,ascf,barboni2019distributed,docaazdaicps,siolcps}. In \cite{adaii}, a framework for cyber attacks and their monitoring methods was proposed. Moreover, limitations on monitoring of the cyber attacks for linear time invariant (LTI) systems were studied. The authors in \cite{ascf} have studied various attack scenarios, such as the replay, zero dynamics, and bias injection attacks for LTI systems based on their proposed framework.

Cyber attacks in MAS and their detection methods have been studied from a system theoretic point of view. Secure consensus tracking control strategies considering two types of attacks were proposed for MAS in \cite{feng2016distributed}. A distributed impulsive control for achieving synchronization in MAS subject to false data injection attacks has also been proposed in \cite{he2018secure}. The work in \cite{ma2016consensus} has suggested a control scheme for multi-agent systems with nonlinearities to reach a consensus while the agents are under deception attacks. In \cite{mustafa2019attack}, cyber-physical attacks on MAS using a system theoretic approach has been studied. It was shown that the attack on one agent can spread into other agents that are reachable from the attacked agent. However, there are limitations and shortcomings in the above work as all cyber attacks on  MAS  are treated as similar to attacks on standard LTI systems. On the other hand, cyber attacks on communication channels among the agents and their significance and impacts have not been addressed and studied in the literature.

For various classes of MAS the controllability conditions are different as discussed in \cite{guan2017controllability,ji2006leader,lou2012controllability,rahmani2009controllability,zhang2013upper}. In \cite{zhang2013upper}, controllability of MAS under undirected network typologies was studied and upper and lower bounds on the controllable subspace of single integrator agents were given in terms of the distance and equitable partitions.

In the present paper, LTI MAS systems having \textit{directed graphs} that are equipped with dynamic output feedback controllers that are also under adversarial false data injection attacks on their communication channels are considered.
Our first objective is to investigate controllability of multi-agent systems from the adversary's point of view. Next, by utilizing the MAS graph topology notions of security controllability indices are introduced. These metrics are utilized in determining from each directly attacked agent how many other agents can be compromised and controlled by the malicious adversary. Finally, conditions under which the adversary is capable of executing zero dynamics cyber attacks on the entire MAS network are provided.

Consequently, the main contributions of this work can be stated as follows. We first introduce the notion of controllability attacks on communication channels of the MAS systems. The importance of these attacks by studying and developing conditions that would provide the adversary full control over the entire MAS system is developed and formalized. Second, it is shown that the adversary is not capable of exciting zero dynamics of the directly attacked and healthy agents simultaneously.

The remainder of the paper is organized as follows. In Section \ref{s:preliminary}, the basic concepts in graph theory that are required are presented and  model of  MAS systems along with their observers are provided. Model of  MAS systems where the communication channels are under attack as well as the objectives of this paper are introduced in Section \ref{s:formulation}. In Section \ref{s:controllability}, necessary and sufficient conditions for the adversary to gain full control over the MAS systems network are formulated and presented. The limitations on zero dynamics attacks that the adversary is capable of injecting by compromising the communication channels are investigated in Section \ref{s:zero}. An illustrative numerical example to demonstrate the capabilities of our proposed methodologies is provided in Section \ref{s:exmple}.

\section{Preliminary}\label{s:preliminary}

\subsection{Graph Theory}
A graph $\mathcal{G}$ with a set of nodes or vertices $\mathcal{V}=\{1, \, 2, ..., \, N\}$, an edge set $\mathcal{E} \subset \mathcal{V} \times \mathcal{V}$, where an edge is defined by the pair of distinct vertices $\mathcal{G}: (i,j) \in \mathcal{E}$, is called directed if $(i,j) \in \mathcal{E}$ does not imply $(j,i) \in \mathcal{E}$. The adjacency matrix of $\mathcal{G}$ is defined as $\mathcal{A}=[a_{ij}] \in \mathbb{R}^{N \times N}$, where $a_{ij}=1$ when there is a link from node $j$ to $i$. The $\mathcal{N}_i$ is the set of neighbors of $i$ which consists of nodes that have an edge to the node $i$, $|\mathcal{N}_i|=d_i$, where $|\cdot|$ denotes the cardinality of the set. The in-degree matrix is defined as $D=\text{diag}(d_1, \, d_2, \dots, \, d_N)$, and the Laplacian matrix is then represented as $L=D-\mathcal{A}$.

\subsection{Model of MAS, Observers, and Consensus Protocols}
The state space representation of a MAS, consisting of $N$ agents, is governed by
\begin{equation}\label{e:agent_i}
\begin{split}
\dot{x}_i (t)&=A x_i (t)+Bu_i(t), \\
y_i(t)&= C x_i (t), \, \, \, \, i=1,...,N,
\end{split}
\end{equation}

\noindent where $x_i (t) \in \mathbb{R}^n$ represents the states of the agent $i$, $u_i (t) \in \mathbb{R}^m$ denotes the control input of the agent $i$, $y_i (t) \in \mathbb{R}^p$  denotes the output of the $i$-th agent, and $t$ denotes time. Matrices $(A, \, B, \, C)$ are of appropriate dimensions. It is assumed that the system $(A, \, B, \, C)$ is controllable and observable.

To design a consensus control protocol for the MAS in \eqref{e:agent_i}, one needs to first estimate the states of the system since only a few are assumed to be measurable. Consider the following observer-based consensus protocol for the system \eqref{e:agent_i} \cite{xu2013consensus}:
\begin{equation}\label{e:obs_i}
\begin{split}
\dot{\hat{x}}_i (t) =& A \hat{x}_i (t)+ B u_i(t)+H \sum_{j \in \mathcal{N}_i}(\zeta_\text{y}(t)+C\zeta_\text{x}(t)), \\
u_i (t)  =&  K \hat{x}_i (t),
\end{split}
\end{equation}

\noindent where $\hat{x}_i(t) \in \mathbb{R}^n$ denotes the state of the observer for the $i$-th agent, {$\zeta_\text{y}(t)=y_j (t)-y_i (t)$, $\zeta_\text{x}(t)=\hat{x}_i (t) -\hat{x}_j (t)$}, $H \in \mathbb{R}^{n \times p}$ is a full column rank observer gain matrix, and $K \in \mathbb{R}^{m\times n}$ is a control gain matrix that should be designed.

\begin{lem}[\cite{horn1991topics}]\label{lem:kroneker}
	Given the matrices $Q,\, W, \, M,$ and $Z$ with appropriate dimensions, the Kronecker product $\otimes$ satisfies the following conditions:
	
	\begin{itemize}
		\item[] $(i)\, (Q+W)\otimes M=Q\otimes M+W\otimes M$;
		\item[] $(ii) \, (Q\otimes W)(M\otimes Z) =(QM)\otimes (WZ).$
	\end{itemize}
	
\end{lem}

\section{Problem Formulation}\label{s:formulation}

\subsection{Cyber Attack on the Communication Links}
As described in \eqref{e:obs_i}, the agent $j \in \mathcal{N}_i$ transmits its observer state $\hat{x}_j(t)$ and output $y_j(t)$ to the agent $i$ as the pair $p_{ji}(t)=(\hat{x}_j(t), \, y_j(t))$. Since this communication is carried out through a network link, it would be prone and vulnerable to cyber attacks, as illustrated in Fig. \ref{fig:attack}. The adversary disguises their injected signals as legitimate information from the neighboring agents of their target such that the targeted agent $i$ only receives the cyber attack signals. This cyber attack can be considered as a man-in-the-middle type of attack \cite{man}.

Consequently, the malicious attacker adds signals $a_1^{ji}(t)=a_{\hat{x}}^{ji} (t)-\hat{x}_j(t)$ and $a_2^{ji}(t)=a_y^{ji} (t)-y_j(t)$ to $p_{ji}(t)$ so that the agent $i$ receives $p_{ji}^\text{a}(t)=(\hat{x}_j(t)+a_1^{ji}(t), \, y_j(t)+a_2^{ji}(t))=(a_{\hat{x}}^{ji} (t), \, a_y^{ji} (t))$ from the agent $j$. Two cyber attack signals $a_{\hat{x}}^{ji} (t) \in \mathbb{R}^n$ and $a_y^{ji} (t) \in \mathbb{R}^p$ are unknown and are to be designed based on the adversary's intentions.

\begin{figure}[!t]
	\centering
	\centerline{\includegraphics[width=\columnwidth]{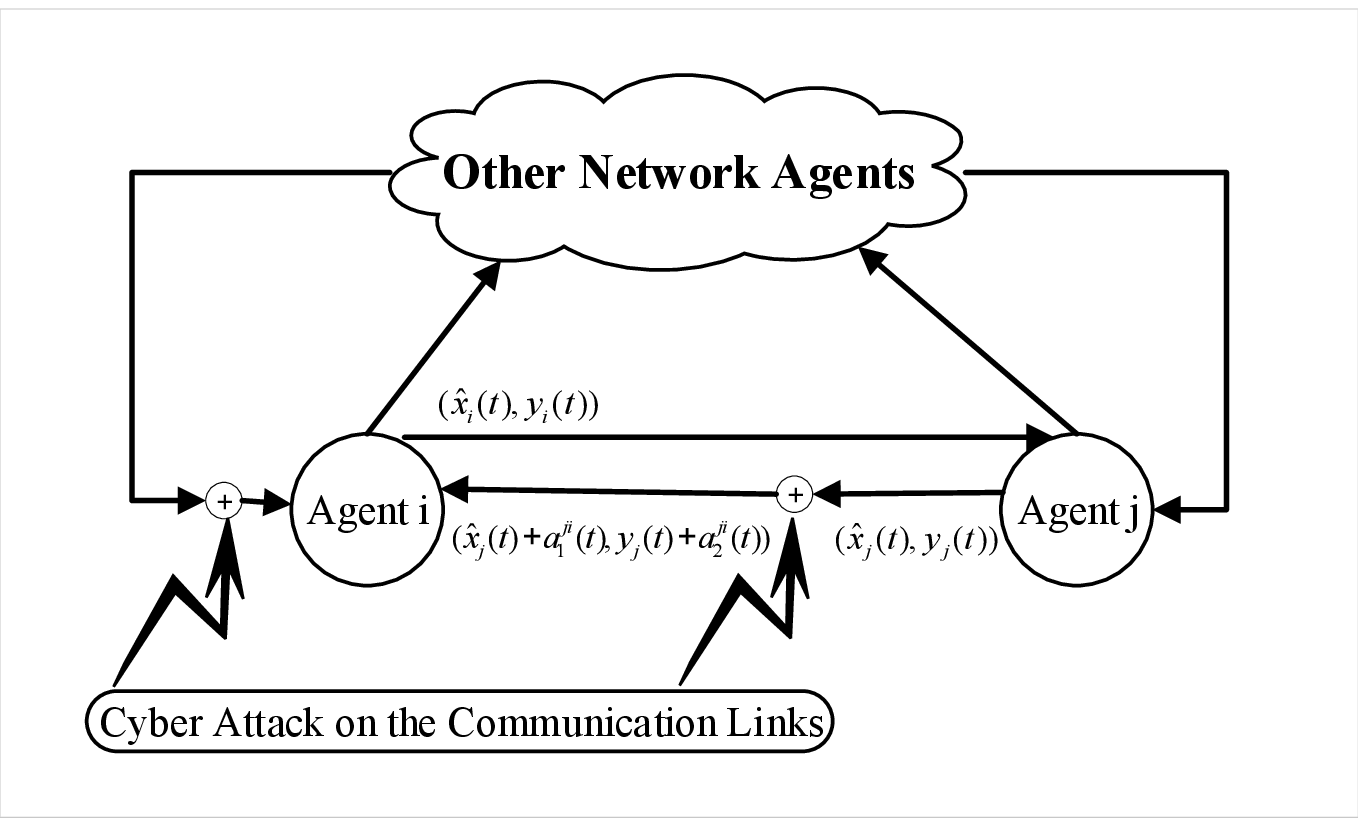}}
	\caption{A communication link cyber attack on the agent $i$. $a_1^{ji}(t)$ designates the cyber attack on the transmitted states of the observer, and $a_2^{ji}(t)$ designates the cyber attack on the transmitted output measurements.}
	\label{fig:attack}
\end{figure}

\begin{assumption}\label{assum:attacks}
	The adversary is capable of executing the worst case scenario attack in which all the incoming communication links of a given agent are under attack.
\end{assumption}

\begin{remark}
	Since the MAS have limited power resources and to make their communications more efficient they use the same communication protocols and encryption/decryption algorithms on all their communication channels \cite{communication}. Hence, if an adversary discovers a vulnerability for one channel of an agent, it is capable of attacking other channels as well.
\end{remark}

Given the observer-based consensus protocol \eqref{e:obs_i}, the closed-loop equations of the system \eqref{e:agent_i} and observer \eqref{e:obs_i} given the communication link cyber attacks can be reformulated as follows:
\begin{IEEEeqnarray}{rCl}
	\dot{x}_i (t)&=& A x_i (t)+BK \hat{x}_i (t), \label{e:agent_i_attack_1}\\
	\dot{\hat{x}}_i (t) &=& A \hat{x}_i (t)+ B K \hat{x}_i (t)+H \sum_{j \in \mathcal{N}_i}(\zeta_\text{y}(t)+q_{i}a_2^{ji}(t) \nonumber \\ &&+C(\zeta_\text{x}(t)-q_i a_1^{ji}(t))) \label{e:obs_i_attack_1},
\end{IEEEeqnarray}

\noindent for $i=1,...,N$ with $q_{i}=1$ if the communication links of the agent $i$ are under attack, and $q_{i}=0$, otherwise.

\subsection{Objectives}
The objectives of this paper are threefold. The first objective is to investigate conditions on the MAS and its Laplacian matrix under which the adversary can gain full controllability over the system in \eqref{e:agent_i_attack_1}. The adversary attempts to directly attack a subset of MAS agents and control the remaining agents as followers of the attacked agents. The second objective is to propose and investigate controllability measures that are based on graph of the MAS which is not fully controllable by the adversary and can be employed to inject attacks on agents that can be controlled through the directly attacked agents. And finally, the third objective is to study the possibility of executing zero dynamics attacks in the MAS governed by \eqref{e:agent_i_attack_1}.

\section{Controllability Cyber Attacks}\label{s:controllability}

\subsection{Conditions for Controllability}
In this subsection, controllability of the MAS \eqref{e:agent_i_attack_1} and its observer that is provided in \eqref{e:obs_i_attack_1} from the adversary's point of view is studied. Let us define
\begin{equation}\label{e:def1}
\begin{split}
\check{A} & = \begin{bmatrix}
A & 0 \\
0 & A
\end{bmatrix}, \, \check{B}=\begin{bmatrix}
0 & B \\
0 & B
\end{bmatrix}, \, \check{H}= \begin{bmatrix}
0 & 0 \\
-H & H
\end{bmatrix}, \\
\check{H}_\text{a} & =\begin{bmatrix}
0 \\
H
\end{bmatrix}, \, \check{K}=\begin{bmatrix}
K & 0 \\
0 & K
\end{bmatrix}, \, \check{C}= \begin{bmatrix}
C & 0 \\
0 & C
\end{bmatrix}.
\end{split}
\end{equation}

\noindent Using \eqref{e:def1}, the augmented dynamic of \eqref{e:agent_i_attack_1} and \eqref{e:obs_i_attack_1} can be derived as follows:
\begin{equation}\label{e:sys_aug_1}
\begin{split}
\dot{\check{x}}_i (t) =& (\check{A} +\check{B}\check{K})\check{x}_i(t)+ \check{H} \check{C} \sum_{j \in \mathcal{N}_i}(\check{x}_i (t)-\check{x}_j (t)) \\
&+\check{H} \check{C} \sum_{j \in \mathcal{N}_i} q_i\check{x}_j (t)+\check{H}_a q_{i}a_{i} (t),
\end{split}
\end{equation}

\noindent where $\check{x}_i(t)= [x_i(t)^\top\, \hat{x}_i(t)^\top]^\top$, and $a_{i} (t)= \sum_{j \in \mathcal{N}_i} a_y^{j i} (t)-Ca_{\hat{x}}^{j i} (t)$.

One can easily partition the agents into two groups, namely the first group contains agents that are directly under attack and the second group consists of agents that receive information from their neighboring agents without any manipulation by the adversary. Consequently, one has $x_\text{f}(t)=[\check{x}_1(t)^\top,\, \check{x}_2(t)^\top, ..., \, \check{x}_{{N}_\text{f}}(t)^\top]^\top$, which designates the state of those agents that are not directly under attack and act as followers. Second, $x_\text{a}(t)=[\check{x}_{N_\text{f} +1}(t)^\top,\, \check{x}_{{N}_\text{f} +2}(t)^\top, ..., \, \check{x}_{N}(t)^\top]^\top$, which designates the directly attacked agents. The subscripts ``$\text{f}$" and ``$\text{a}$" are used to denote followers and attacked agents, respectively. ${{N}}_\text{f}$ denotes the number of followers and ${{{N}}_\text{a}}$ denotes the number of attacked agents, where ${{N}}={{N}}_\text{f}+{{N}}_\text{a}$. Without loss of generality, we assume that the first ${N}_\text{f}$ agents are not under attack. Consequently, the Laplacian matrix can be partitioned into the following form:
\begin{equation}\label{e:lap}
L=\begin{bmatrix}
L_\text{f} & l_{\text{fa}} \\
l_{\text{af}} & L_\text{a}
\end{bmatrix},
\end{equation}

\noindent where $L_\text{f} \in \mathbb{R}^{N_\text{f} \times N_\text{f}}$ is a grounded Laplacian matrix \cite{pirani2014spectral}, $L_\text{a} \in \mathbb{R}^{N_\text{a} \times N_\text{a}}$, $l_{\text{fa}} \in \mathbb{R}^{N_\text{f} \times N_\text{a}}$, and $l_{\text{fa}} \in \mathbb{R}^{N_\text{a} \times N_\text{f}}$.

The dynamics of all $N$ agents can be expressed as follows:
\begin{IEEEeqnarray}{rCl}
	\dot{x}_\text{a}(t) &=& A_\text{a} x_\text{a}(t)+B_\text{a} a (t), \label{e:sys1_attacked} \\
	\dot{x}_\text{f}(t) &=& A_\text{f} x_\text{f}(t) +A_\text{fa} x_\text{a}(t), \label{e:sys1_followers}
\end{IEEEeqnarray}

\noindent where $A_\text{a} =I_{N_\text{a}} \otimes (\check{A}+\check{B}\check{K})+D_\text{a} \otimes \check{H} \check{C}$, $A_\text{f}=I_{N_\text{f}} \otimes (\check{A}+\check{B}\check{K})+L_\text{f}\otimes \check{H}\check{C}$, $D_\text{a}=\text{diag}(d_{N_\text{f}+1},d_{N_\text{f}+2}, ... ,\, d_{N})$, $B_\text{a}=I_{N_\text{a}} \otimes \check{H}_\text{a}$, $A_\text{fa}=l_{\text{fa}}\otimes \check{H}\check{C}$, and $a (t)=[{a_{N_\text{f} +1}}(t)^\top,\, {a_{N_\text{f} +2}}(t)^\top, ..., \, {a_{N}}(t)^\top]^\top$. The dynamics of directly attacked agents \eqref{e:sys1_attacked} and the followers \eqref{e:sys1_followers} can be augmented in the following form:
\begin{equation}\label{e:sys}
\begin{bmatrix}
\dot{x}_\text{a}(t) \\
\dot{x}_\text{f}(t)
\end{bmatrix}= \begin{bmatrix}
A_\text{a} & 0 \\
A_\text{fa} & A_\text{f}
\end{bmatrix} \begin{bmatrix}
{x}_\text{a}(t) \\
{x}_\text{f}(t)
\end{bmatrix}+ \begin{bmatrix}
B_\text{a} \\
0
\end{bmatrix} a(t) .
\end{equation}

We are now in a position to state our first definition as well as the first result of this section.
\begin{definition}\label{def:control}
	The MAS described in \eqref{e:sys} is controllable by the adversary if, for every $x_\text{a}^*$ and $x_\text{f}^*$ and every finite $T>0$, there exists an attack signal $a(t)$, $0<t<T$, such that the MAS states do transition from $x_\text{a}(0)=0$ and $x_\text{f}(0)=0$ to $x_\text{a}(T)=x_\text{a}^*$ and $x_\text{f}(T)=x_\text{f}^*$, respectively.
\end{definition}

\begin{theorem}\label{th:iff}
	The adversary is capable of controlling the system \eqref{e:sys} if the pairs $(\check{A}+\check{B}\check{K}+d_i \check{H} \check{C},\check{H}_\text{a})$ for $i=N_{\text{f}}+1,..., \, N$, and $( A_\text{f},A_\text{fa})$ are controllable, $\text{rank}(\sum_{k=1}^{2nN-1}M_k Q_k^\top w_1)$ is either equal to $ 2nN_\text{f}$ if $N_\text{f} \leq N_\text{a}$ or equal to $ 2nN_\text{a}$ if $N_\text{a} < N_\text{f}$, where $Q=[Q_1, \dots , Q_{(2n N-1)}]$, $Q_k=A_\text{a}^k B_\text{a}$ for $k=1,\dots, \, 2n N-1$, $Q_0=B_\text{a}$, $M=[M_1, \dots, M_{(2n N-1)}]$, $M_k=\sum_{z=0}^{k-1} A_\text{f}^z A_\text{fa} Q_{k-1-z}$, columns of $M_k$ are nonzero, and $w_1 \in \mathbb{R}^{(2n N_\text{a})\times (2n N_\text{a})}$ is a matrix that satisfies $w_1 \in \text{ker}(B_\text{a}^\top)$.	
\end{theorem}

\begin{proof}
Let us define
	\begin{equation}\label{e:def:giant}
	A^*=\begin{bmatrix}
	A_\text{a} & 0 \\
	A_\text{fa} & A_\text{f}\end{bmatrix}, \, B^*=\begin{bmatrix}
	B_\text{a} \\
	0
	\end{bmatrix}.
	\end{equation}
	The controllability matrix of the system \eqref{e:sys}, $\mathcal{C}^*=	\begin{bmatrix}
	B^*, & A^* B^*, & \dots & {(A^*)}^{2n N-1} B^*
	\end{bmatrix}$, can be expressed in the following form:
	\begin{equation*}
	\mathcal{C}^* =	\begin{bmatrix}
	Q_0 & Q_1 & \cdots & Q_{(2n N-1)} \\
	0 & M_1 & \cdots & M_{(2n N-1)}
	\end{bmatrix}.
	\end{equation*}
	
	\noindent For the system \eqref{e:sys} to be controllable, $\mathcal{C}^*$ should be of full row rank. Hence, controllability is achieved if $[Q_0, \, Q]$ and $M$ are right invertible and rows of $Q$ and $ M$, under some conditions that are provided below, are linearly independent.
	
	From the definition of $Q_0$ and $Q$, one can conclude the right invertibility of $[Q_0, \, Q]$ is equivalent to the pair $(A_\text{a}, B_\text{a})$ being controllable. For this pair the matrix $D_\text{a}$ is diagonal, therefore, $A_\text{a}=\text{blockdiag}((\check{A}+\check{B}\check{K}+d_{N_\text{f}+1} \check{H} \check{C}),..., \, (\check{A}+\check{B}\check{K}+d_{N} \check{H} \check{C}))$ is a block diagonal matrix. The operator $\text{blockdiag}(\cdot)$ denotes a block diagonal matrix. In addition, $I_{N_\text{a}} \otimes \check{H}_\text{a}=\text{bockdiag}(\check{H}_\text{a},..., \, \check{H}_\text{a})$ is block diagonal. Hence, the controllability condition can be studied for each attacked agent separately.
	
	The matrix $M=[M_1, \dots, M_{(2n N-1)}]$ can be written as the product of two matrices, namely $M^*$ and $Q^*$, i.e., $M=M^* Q^*$, where
	\begin{equation*}
	\begin{split}
	M^*= & \begin{bmatrix}
	A_\text{fa} & A_\text{f} A_\text{fa} & \dots & (A_\text{f})^{2n N-2} A_\text{fa}
	\end{bmatrix}, \, \\
	Q^*= &\begin{bmatrix}
	B_\text{a} & A_\text{a} B_\text{a} & A_\text{a}^2 B_\text{a} & \cdots & A_\text{a}^{(2n N-2)} B_\text{a} \\
	0 & B_\text{a} & A_\text{a} B_\text{a} & \cdots & A_\text{a}^{(2n N-3)}B_\text{a} \\
	\vdots & \vdots & \ddots & & \vdots \\
	\vdots & \vdots &   & B_\text{a} & A_\text{a} B_\text{a} \\
	0 & 0 & \cdots & 0 & B_\text{a}
	\end{bmatrix}
	\end{split}.
	\end{equation*}
	
	\noindent The rows of the matrices $M_k=\sum_{z=0}^{k-1} A_\text{f}^z A_\text{fa} Q_{k-1-z}$, $k=1, \dots,2n N-1$, are equal to the rows of $M^*$ multiplied by the columns of $Q^*$. The matrices $M_k$ not having any zero column is equivalent to them not having any basis of $\text{ker}(M^*)$ in common with basis of $\text{Im}(Q^*)$. In other words, $\text{ker}(M^*) \cap \text{Im}(Q^*)={0}$. This condition along with the fact that the number of rows of $M^*$ is smaller than the dimensions of $Q^*$, in turn imply that $\text{rank}(M)=\text{rank}(M^*)$. Consequently, for $M$ to be right invertible, $M^*$ should be of full row rank, which is satisfied if the pair $(A_\text{f}, A_\text{fa})$ is controllable.
	
	Considering $w_1 \in \text{ker}(B_\text{a}^\top)$ and an appropriate matrix $w_2$, one has
	\begin{equation}\label{e:QM}
	\begin{bmatrix}
	w_1^\top & 0 \\
	0 & w_2^\top
	\end{bmatrix}\begin{bmatrix}
	B_\text{a} & Q \\
	0 & M
	\end{bmatrix} = \begin{bmatrix}
	0 & w_1^\top Q \\
	0 & w_2^\top M
	\end{bmatrix}.
	\end{equation}
	\noindent Rows of $w_1^\top Q$ and $w_2^\top M$ should not be linearly dependent to have a right invertible $\mathcal{C}^*$. This is satisfied if $w_1^\top Q \neq w_2^\top M$ for every $w_2$. This implies there does not exist any $w_2^\top$ such that the rows of $w_1^\top Q$ and $w_2^\top M$ are linearly dependent if $\text{ker}(M) \nsubseteq \text{ker}(w_1^\top Q)$. This condition is satisfied if $\text{ker}(M) \cap \text{ker}( w_1^\top Q)=0$. For the latter condition to be satisfied the following matrix
	\begin{equation}
	\begin{split}
	S= & \begin{bmatrix}
	M_1 & \dots & M_{(2n N-1)}
	\end{bmatrix} \times \begin{bmatrix}
	Q_1^\top w_1 \\
	\vdots \\
	Q_{(2n N-1)}^\top w_1
	\end{bmatrix} \\
	=& \sum_{k=1}^{2nN-1}M_k Q_k^\top w_1 ,
	\end{split}
	\end{equation}
	
	\noindent should be either full row rank if $N_\text{f} \leq N_\text{a}$ or full column rank if $N_\text{a} < N_\text{f}$. This completes the proof of the theorem.
\end{proof}

\begin{remark}
	Since the conditions in Theorem \ref{th:iff} are difficult to verify the adversary may not be able to gain control over the entire network as described in Definition \ref{def:control}. In such a scenario the adversary is capable of injecting its attack signals to the directly targeted agents and control the followers through them. In this type of attack, states of the directly attacked agents are used as  control inputs to the followers.
\end{remark}

The definition of controllability of MAS  followers is provided below.

\begin{definition}\label{def:control2}
	The followers \eqref{e:sys1_followers} are controllable through the directly attacked agents by the adversary if for every $x_\text{f}^*$ and every finite $T>0$, there exists a proper $x_\text{a}(t)$, $0<t<T$, such that the state transitions can be accomplished from $x_\text{f}(0)=0$ to $x_\text{f}(T)=x_\text{f}^*$.
\end{definition}

\begin{assumption}\label{assum:1}
	The set of eigenvectors of $L_\text{f}$ span $\mathbb{R}^{N_\text{f}}$.
\end{assumption}

It should be noted that the grounded Laplacian matrix in case of a directed graph is not necessarily diagonalizable. For example, consider the Laplacian matrix $L=[$$1$, $0$, $0$, $-1$; $-1$, $1$, $0$, $0$; $0$, $-1$, $1$, $0$; $0$, $0$, $-1$, $1$$]$ and its corresponding grounded Laplacian matrix $L_\text{f}=[$$1$, $0$, $0$; $-1$, $1$, $0$; $0$, $-1$, $1$$]$, where the agent 4 is directly under cyber attack. The algebraic multiplicity of the eigenvalue of $L_\text{f}$, namely $\lambda=1$, is $3$, however, its geometric multiplicity is $1$,  implying that $L_\text{f}$ is not diagonalizable. Since in the Theorem \ref{th:conses1} provided below one requires $L_\text{f}$ to be diagonalizable, the above Assumption \ref{assum:1} is given.
\begin{proposition}[\cite{brogan1991modern}]\label{pro:1}
	The system $\dot{x}_i(t)=Ax_i (t)+Bu_i(t)$ is controllable if and only if $\forall$ $v_k$, $k=1,..., \, n$, where $v_k$ is the $k$-th eigenvector of $A$, $v_k \notin \text{ker}(B^T)$.
\end{proposition}

\begin{theorem}\label{th:conses1}
	Under Assumption \ref{assum:1} the adversary is capable of controlling the system \eqref{e:sys1_followers} through the directly attacked agents \eqref{e:sys1_attacked} according to the Definition \ref{def:control2} if and only if the pairs $(\check{A}+\check{B}\check{K}+d_i \check{H} \check{C},\check{H}_\text{a})$,  $(L_\text{f},l_{\text{fa}})$, and $(\check{A}+\check{B}\check{K}+\lambda_j \check{H} \check{C},\check{H}\check{C})$ are controllable for $i=N_{\text{f}}+1,..., \, N$ and $j=1,..., \, N_\text{f}$, where $\lambda_j$ is the $j$th eigenvalue of $L_f$.
\end{theorem}

\begin{proof}
	In this attack scenario, the adversary uses $x_\text{a}(t)$ as the control input to the followers. Hence, the adversary should be capable of setting $x_\text{a}(t)$ to its desired value, which can be achieved if \eqref{e:sys1_attacked} is controllable. Consequently, the followers in \eqref{e:sys1_followers} should be controllable through $x_\text{a}(t)$. Since \eqref{e:sys1_attacked} is considered to be controllable, the adversary is capable of designing $a(t)$ such that $x_\text{a}(t)$ tracks its desired trajectory (see Theorem 5.2.5 and Corollary 5.2.6 in \cite{willems1997introduction}).

	The controllability condition of the pair $(A_\text{a},B_\text{a})$ was studied in Theorem \ref{th:iff}. Controllability of $( A_\text{f}, A_{\text{fa}})$ indicates that the followers are controlled via the state of the attacked agents, $x_\text{a}(t)$. In view of Assumption \ref{assum:1}, there always exists an invertible matrix $P$, with its rows representing $N_\text{f}$ right eigenvectors of $L_\text{f}$, such that $PL_\text{f} P^{-1}=\text{diag}(\lambda_1,...,\lambda_{N_\text{f}})$. Using the similarity transformation $P \otimes I_n$, \eqref{e:sys1_followers} can be rewritten as
	\begin{equation}\label{e:th2}
	\dot{x}^\text{p}_\text{f} (t)=(P \otimes I_n) A_\text{f} (P^{-1} \otimes I_n) x_\text{f}^\text{p} (t)+(P \otimes I_n)(l_{\text{fa}}\otimes \check{H}\check{C})x_\text{a}(t),
	\end{equation}
	
	\noindent where $x_\text{f}^\text{p} (t)=(P \otimes I_n)x_\text{f}(t)$. Since $P$ is nonsingular, the controllability of $x_\text{f}^\text{p} (t)$ implies the controllability of $x_\text{f} (t)$. The matrix $(P \otimes I_n) A_\text{f} (P^{-1} \otimes I_n)=\text{blockdiag}(\check{A}+\check{B}\check{K}+\lambda_1 \check{H} \check{C},..., \, \check{A}+\check{B}\check{K}+\lambda_{N_\text{f}} \check{H} \check{C})$ is block diagonal and $(P \otimes I_n)(l_{\text{fa}}\otimes \check{H}\check{C})=Pl_{\text{fa}} \otimes \check{H}\check{C}$. Consequently, \eqref{e:th2} can be expressed in the following form:
	\begin{equation}\label{e:th2:2}
	\begin{split}
	\dot{x}^p_\text{f} (t)= &\text{blockdiag}(\check{A}+\check{B}\check{K}+\lambda_1 \check{H} \check{C},..., \, \check{A}+\check{B}\check{K}+\lambda_{N_\text{f}} \check{H} \check{C}) \\
	&\times  x_\text{f}^\text{p} (t)+(P l_{\text{fa}}\otimes \check{H}\check{C})x_\text{a}(t).
	\end{split}
	\end{equation}
	
	\noindent Since the rows of $P$ are the right eigenvectors of $L_\text{f}$, in view of Proposition \ref{pro:1}, the controllability of $(L_\text{f}, \, l_{\text{fa}})$ can be interpreted as not having completely zero rows in the matrix $P l_{\text{fa}}$. The vector $x_\text{f}(t)$ contains the states of $N_\text{f}$ followers, however, due to the similarity transformation, $x_\text{f}^\text{p} (t)$ contains a combination of these states, but still one has $N_\text{f}$ modes that are the $N_\text{f}$ different blocks of $\text{blockdiag}(\check{A}+\check{B}\check{K}+\lambda_1 \check{H} \check{C},..., \, \check{A}+\check{B}\check{K}+\lambda_{N_\text{f}} \check{H} \check{C})$. Next we provide and prove the necessary and sufficient conditions of our proposed methodology that are stated in this theorem.

	\textit{\textbf{Necessary Condition}}: Assume the $j$-th mode of \eqref{e:th2:2} is controllable through $x_\text{a}(t)$, while either $(\check{A}+\check{B}\check{K}+\lambda_j \check{H} \check{C},\check{H}\check{C})$ is not controllable or the $j$-th row of $P l_{\text{fa}}$ is zero. Due to block diagonal structure of \eqref{e:th2:2}, either the uncontrollability of $(\check{A}+\check{B}\check{K}+\lambda_j \check{H} \check{C},\check{H}\check{C})$ or the $j$-th row of $P l_{\text{fa}}$ being zero results in the uncontrollability of the mode $j$, which contradicts the assumption on this mode.

	\textit{\textbf{Sufficient Condition}}: Suppose that the mode $j$ is uncontrollable, while $(\check{A}+\check{B}\check{K}+\lambda_j \check{H} \check{C},\check{H}\check{C})$ is controllable and the $j$-th row of $P l_{\text{fa}}$ is nonzero. However, from the block diagonal structure of \eqref{e:th2:2}, the mode $j$ being uncontrollable implies that either $(\check{A}+\check{B}\check{K}+\lambda_j \check{H} \check{C},\check{H}\check{C})$ is uncontrollable or the $j$-th row of $P l_{\text{fa}}$ is zero, which is a contradiction. This completes the proof of the theorem.	
\end{proof}

\begin{remark}
	As shown in Theorem \ref{th:conses1}, the problem of interest here is to show that there exists a proper $x_\text{a}(t)$ that satisfies the controllability objective provided in Definition \ref{def:control2}. However, designing the attack signal $a(t)$ such that $x_\text{a}(t)$ follows the adversary's desired trajectory is not within the scope of this paper and is not addressed here.
\end{remark}

\begin{remark}
	Generally speaking, the difference between the goals in Definitions \ref{def:control} and \ref{def:control2} has resulted in different types of conditions that need to be satisfied in Theorems \ref{th:iff} and \ref{th:conses1}. In Theorem \ref{th:iff}, the conditions are more restrictive, however they ensure  controllability over the entire network for the adversary. Nevertheless, the main objective of the malicious hacker is to exert the maximum possible influence on the MAS given the available resources. Consequently, the adversary may not be able to control the entire network as studied in Theorem \ref{th:iff}, whereas they can still compromise the system and lead the MAS to dangerous trajectories only if the conditions in Theorem \ref{th:conses1} are satisfied. This result is illustrated through the numerical example that is provided in Section \ref{s:exmple}.
\end{remark}

\subsection{Cyber Security Controllability Index}
As shown in Theorem \ref{th:conses1}, the only condition on controllability of the MAS that connects the structure of the communication graph among the followers and the directly attacked agents is the controllability of $(L_\text{f},l_{\text{fa}})$. By leveraging this controllability condition, we aim to define two security metrics for the MAS. These notions can be used to evaluate the security of the MAS with respect to their controllability by an adversarial intruder. In this subsection, we assume that all the conditions in Theorem \ref{th:conses1}, except for the controllability of $(L_\text{f},l_{\text{fa}})$, hold true. Let us denote $\hat{L}_\text{f}=PL_\text{f} P^{-1}=\text{diag}(\lambda_1,...,\lambda_{N_\text{f}})$ and $\hat{l}_\text{fa}=Pl_\text{fa}$, where rows of $P$ are the $N_\text{f}$ right eigenvectors of $L_\text{f}$.

\begin{definition}\label{def:SCIi}
	The security controllability index of the directly attacked agent $i$, designated by $SCI_i$, is defined by:
	\begin{equation}\label{e:SCIi}
	SCI_i=\text{rank}(\mathcal{C}_i), \, \, \, \, i=N_\text{f}+1,..., \, N,
	\end{equation}
	
	\noindent where $\mathcal{C}_i=[(\hat{l}_\text{fa})_i, \, \hat{L}_\text{f} (\hat{l}_\text{fa})_i, \dots, \, \hat{L}_\text{f}^{N_\text{f}-1} (\hat{l}_\text{fa})_i]$ denotes the controllability matrix (considered not to be ill-conditioned) and $(\hat{l}_\text{fa})_i$ denotes the $i$-th column of $\hat{l}_\text{fa}$.
\end{definition}

The maximum value for $SCI_i$ can be $N_\text{f}$, which if satisfied implies that all the followers can be manipulated and controlled via the agent $i$.

\begin{definition}\label{def:SCI}
	The security controllability index (SCI) of the MAS is defined as
	\begin{equation}\label{e:SCI}
	SCI=\text{rank}(\mathcal{C}),
	\end{equation}
	
	\noindent	where $\mathcal{C}=[\hat{l}_\text{fa}, \, \hat{L}_\text{f} \hat{l}_\text{fa}, \dots, \, \hat{L}_\text{f}^{N_\text{f}-1} \hat{l}_\text{fa}]$.
\end{definition}

The problem for the adversary is to find the minimum number of directly attacked agents that gives the full control over the multi-agent network. More specifically, the adversary's goal is to $\text{minimize} \, |N_\text{a}|$ such that $SCI=N_\text{f}$. In the literature this problem is referred to as actuator placement problem \cite{actuatorplacement}. Solving the above minimization problem provides the adversary with the minimum required number of agents that the hacker needs to compromise and attack. A few methods that incorporate graph of the network to select agents for ensuring controllability over the MAS have been suggested in \cite{rahmani2009controllability,ji2006leader,slotine,leader1,leader2}.

\begin{remark}\label{rem:index}
	Due to the possibility of existence of sufficiently small singular values and ill-conditioning of the matrices
$\mathcal{C}$ and $\mathcal{C}_i$ for $i=N_\text{f}+1,..., \, N$ in \eqref{e:SCIi} and \eqref{e:SCI}, one may have nearly singular matrices. In such cases  $\text{rank}(\mathcal{C})$ and $\text{rank}(\mathcal{C}_i)$  can be computed by imposing a tolerance condition on computation of the rank such that if the singular value is smaller than a pre-specified tolerance level it is then considered to be  zero.
\end{remark}

\section{Zero Dynamics Attacks Through the Communication Links}\label{s:zero}
Given an $s=s_\text{a}$ and the dynamics of the directly attacked agents in \eqref{e:sys1_attacked}, the zero dynamics of the MAS are those $s_\text{a}$ in which the Rosenbrock system matrix
\begin{equation}\label{e:sysmtrx_a}
P_\text{a}(s)=\begin{bmatrix}
sI-A_\text{a} & -(I_{N_\text{a}} \otimes \check{H}_\text{a}) \\
I_{N_\text{a}} \otimes C & 0
\end{bmatrix}
\end{equation}

\noindent is rank deficient, i.e., its rank falls below its normal rank. This implies that there exist nonzero $x_{\text{a}0}$ and $a_0$ such that
\begin{equation}\label{e:zero_a}
\begin{bmatrix}
s_\text{a}I-A_\text{a} & -(I_{N_\text{a}} \otimes \check{H}_\text{a}) \\
I_{N_\text{a}} \otimes \check{C} & 0
\end{bmatrix} \begin{bmatrix}
x_{\text{a}0} \\
a_0
\end{bmatrix}=0,
\end{equation}

\noindent where $X_\text{a}(t)=x_{\text{a}0} e^{s_\text{a} t}$ with $X_\text{a}(t)$ defined as the solution to \eqref{e:sys1_attacked} and $a(t)=a_{0} e^{s_\text{a} t}$.

The zero dynamics of the followers \eqref{e:sys1_followers} are defined as $s=s_\text{f}$ and are associated with nonzero directional vectors $x_{\text{f}0}$ and $x_{\text{af}}$ such that the following is satisfied:
\begin{equation}\label{e:zero_f}
\begin{bmatrix}
s_\text{f}I-A_\text{f} & -(l_{\text{fa}}\otimes \check{H}\check{C}) \\
I_{N_\text{f}} \otimes \check{C} & 0
\end{bmatrix} \begin{bmatrix}
x_{\text{f}0} \\
x_{\text{af}}
\end{bmatrix}=0,
\end{equation}

\noindent where $X_\text{f}(t)=x_{\text{f}0} e^{s_\text{f} t}$ with $X_\text{f}(t)$ as the solution to \eqref{e:sys1_followers} and $X_\text{a}(t)=x_{\text{af}} e^{s_\text{f} t}$.

\begin{definition}\label{def:zero}
	The zero dynamics $s_\text{a}$ and $s_\text{f}$ are excited in the systems \eqref{e:sys1_attacked} and \eqref{e:sys1_followers} if their initial conditions and the attack signal satisfy the conditions in \eqref{e:zero_a} and \eqref{e:zero_f}, respectively.
\end{definition}

From \eqref{e:zero_a} and \eqref{e:zero_f} one can conclude that the differences that exist between the attacked agents and the followers can result in having different zero dynamics in these two groups. Moreover, in case of an attacker exciting the zero dynamics, the states should satisfy $x_i(t) \in \text{ker}(C)$, for $i=1, \dots, \, N$, to have a zero output in the system \cite{siolcps}.

\begin{lem}\label{lem:zero_excite}
	The zero dynamics of the followers \eqref{e:sys1_followers} and the directly attacked agents \eqref{e:sys1_attacked} are excited by the adversary in the sense of Definition \ref{def:zero} if \eqref{e:zero_f} and \eqref{e:zero_a} for $s_\text{a}=s_\text{f}$ hold true, while $(l_{\text{fa}}\otimes \check{H}\check{C})x_{\text{af}} \ne 0$ and $(I_{N_\text{a}} \otimes \check{H}_\text{a})a_{0} \ne 0$, respectively.
\end{lem}

\begin{proof}
	Suppose the output of the system is zero with nonzero $x_{\text{af}}$ and $a_0$ that are in  the $\text{ker}(l_{\text{fa}}\otimes \check{H}\check{C})$ and $\text{ker}(I_{N_\text{a}} \otimes \check{H}_\text{a})$, respectively. This implies that the attack signal does not have an impact on exciting the zero dynamics. Therefore, in  case of zero dynamics attack by the adversary it is necessary for the attack signals to satisfy $x_{\text{af}} \notin \text{ker}(l_{\text{fa}}\otimes \check{H}\check{C})$ and $a_0 \notin \text{ker}((I_{N_\text{a}} \otimes \check{H}_\text{a}))$. This completes the proof of the lemma.
\end{proof}

\begin{theorem}\label{th:zero}
	The adversary is not capable of simultaneously exciting the zero dynamics of the directly attacked agents \eqref{e:sys1_attacked} and the followers \eqref{e:sys1_followers} in the sense of Definition \ref{def:zero}.
\end{theorem}

\begin{proof}
	Suppose the adversary excites the zero dynamics of directly attacked agents in \eqref{e:sys1_attacked} so that $X_\text{a}(t)=x_{\text{a}0} e^{s_\text{a} t}$ is in $\text{ker}(I_{N_\text{a}} \otimes \check{C})$. Consequently, $x_\text{a0}$ should be of the form $x_\text{a0}=I_{N_\text{a}} \otimes \check{x}_\text{a0}$, where $\check{x}_\text{a0} \in \text{ker}(\check{C})$. Since $(l_{\text{fa}}\otimes \check{H}\check{C})\times (I_{N_\text{a}} \otimes \check{x}_\text{a0})=l_{\text{fa}}\otimes \check{H}\check{C} \check{x}_\text{a0}=0$, one can conclude that $x_{\text{af}}=X_\text{a}(0)=x_\text{a0} \in \text{ker}(l_{\text{fa}}\otimes \check{H}\check{C})$, which according to Lemma \ref{lem:zero_excite} implies that the adversary is not capable of exciting the zero dynamics of the followers. Now let us assume \eqref{e:zero_f} holds and the zero dynamics of the followers are excited by the adversary. Therefore, $(l_{\text{fa}}\otimes \check{H}\check{C})x_{\text{af}} \ne 0$ is satisfied. This implies that $X_\text{a}(0)=x_{\text{af}} \notin \text{ker}(I_{N_\text{a}} \otimes \check{C})$. Hence, \eqref{e:zero_a} does not hold and the zero dynamics of the directly attacked agents \eqref{e:sys1_attacked} cannot be excited by the adversary. This completes the proof of the theorem.
\end{proof}

\section{Numerical Example}\label{s:exmple}
In this numerical example, the controllability conditions that are provided in Theorems \ref{th:iff} and \ref{th:conses1} are studied for a MAS system consisting of 6 agents. The agent dynamics  and its observer are given by \eqref{e:agent_i}, and \eqref{e:obs_i}, respectively, with the following matrices \cite{li2009consensus}:
\begin{equation*}
\begin{split}
A & =\begin{bmatrix}
-2 & 2 \\
-1 & 1
\end{bmatrix}, \, B=\begin{bmatrix}
1 \\
0
\end{bmatrix}, \, C=\begin{bmatrix}
1 & 0 \\
0 & 1
\end{bmatrix}, \\
H &= \begin{bmatrix}
0 & 0.3 \\
-0.3 & 0
\end{bmatrix}, \, K=\begin{bmatrix}
-1 & 2
\end{bmatrix}.
\end{split}
\end{equation*}

The communication graph among the agents is shown in Fig. \ref{fig:graph}, and its corresponding Laplacian matrix is $L=[$1, 0, 0, 0, 0, -1; 0, 2, 0, 0, -1, -1; 0, -1, 1, 0, 0, 0; 0, -1, -1, 2, 0, 0; 0, -1, 0, 0, 1, 0; 0, 0, 0, 0, -1, 1$]$.

\begin{figure}[!t]
	\centering
	\centerline{\includegraphics[width=1.3in]{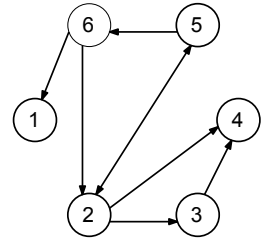}}
	\caption{Communication graph of the MAS system.}
	\label{fig:graph}
\end{figure}

Let us assume that the incoming communication links of agents 4, $5$, and 6 are under  attack so that one obtains $L_\text{f}=[$1, 0, 0; 0, 2, 0; 0, -1, 1$]$ and $l_{\text{fa}}=[$0, 0, -1; 0, -1, -1; 0, 0, 0$]$, where the eigenvalues of $L_\text{f}$ are $\lambda_1=1$, $\lambda_2=1$, $\lambda_3=2$, corresponding to the right eigenvectors $[1,\, 0, \, 0]^T$, $[0, \, 0, \, 1]^T$, and $[0, \, 0.7071, \, -0.7071]^T$, respectively. Since the geometric multiplicity of each eigenvalue is  equal to its algebraic multiplicity, conditions in Assumption \ref{assum:1} hold. In this example, the conditions in Theorem \ref{th:iff} are not satisfied and $\text{rank}(\mathcal{C}^*)=5$. Hence, the adversary does not have control over the entire MAS system as provided in Definition \ref{def:control}, however, an adversary may still impact the followers as described in Definition \ref{def:control2} and  Theorem \ref{th:conses1}.

Considering Theorem \ref{th:conses1},  the rank of the controllability matrices $(\check{A}+\check{B}\check{K}+d_i \check{H} \check{C},\check{H}_a)$ for $i=4,\, 5, \, 6$ are equal to 4, the rank of the controllability of the pair $(L_\text{f},l_{\text{fa}})$ is 3, and the rank of the controllability  of $(\check{A}+\check{B}\check{K}+\lambda_j \check{H} \check{C},\check{H}\check{C})$ for $j=1, \, 2, \, 3$ is equal to 4. Therefore, the adversary has the capability of manipulating and controlling the three agents 1, 2, 3 by simultaneously attacking the agents 4, 5, and 6.

As shown in Fig. \ref{fig:x}, the six agents reach a consensus and their states converge, while at $t=30$ (s) the adversary injects its attack signals to the agents 4, 5, and 6 and the remaining agents are controlled through the directly attacked agents. In Fig. \ref{fig:x}, to illustrate the capability of the adversary in controlling all the agents, the states of each agent are set to different values by choosing different attack signals for the directly attacked agents.

In Fig. \ref{fig:xCons}, it can be seen that the  attack that has occurred at $t=30$ (s) is designed such that the agents reach a new consensus that is desirable to the adversary. In this attack scenario, the directly attacked agents have the same attack signals so that they reach to the same point and the remaining agents follow them. This example illustrates that even without having full controllability over the MAS systems as stated in Definition \ref{def:control}, the adversary is capable of imposing a major impact on the trajectory and behavior of the agents.

\begin{figure}[!t]
	\centering
	\centerline{\includegraphics[width=\columnwidth]{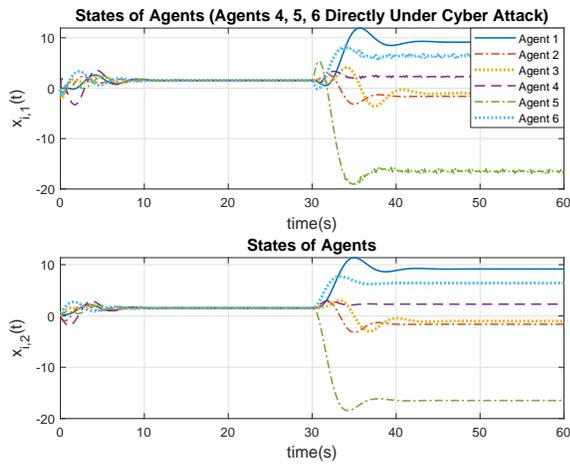}}
	\caption{State trajectories of the six agents in presence of cyber attack injected at $t=30$ (s).}
	\label{fig:x}
\end{figure}

\begin{figure}[!t]
	\centering
	\centerline{\includegraphics[width=\columnwidth]{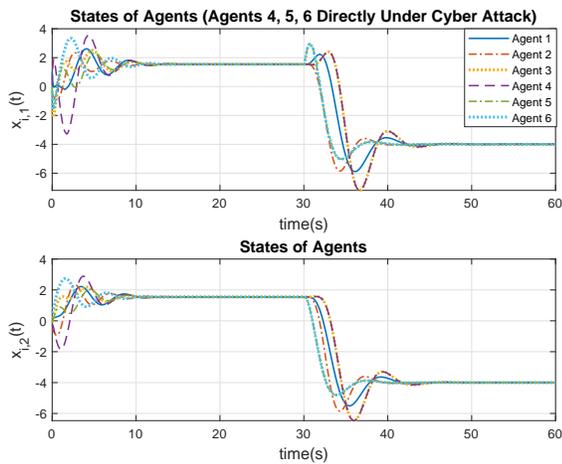}}
	\caption{Change in the consensus set point by the adversary injecting cyber attack signals  at $t=30$ (s).}
	\label{fig:xCons}
\end{figure}

The security controllability index for the directly attacked agents are $SCI_4=0$, $SCI_5=1$, $SCI_6=2$, and for the MAS system is $SCI=3$. It follows that $SCI_4=0$, which implies that through agent 4 the adversary is not capable of controlling any of the followers. However, attacking the agents 5 and 6 do not provide controllability over the agent 4 to the adversary, and hence, in this case it will be attacked directly.

\section{Conclusion}
In this paper, certain types of  cyber attacks on MAS systems were investigated and developed. In one cyber attack scenario, the adversary targets the incoming communication links for a team of agents and disguises the attack signals as transmitted information among the agents. Therefore, there are two groups of agents, those that are first directly attacked, and those that can be considered as followers of the first group. The conditions under which the adversary has full control over the two agent groups were investigated. The notions of security controllability for each of the directly attacked agents as well as the entire MAS system have been proposed and developed. These notions can be used to identify agents that allow the adversary high control authority over the MAS network. Finally, it was shown that the adversary is not capable of simultaneously performing zero dynamics attacks on the directly attacked agents and the followers. Relaxing the assumption in Theorem \ref{th:conses1} is challenging and is a focus of our future work. Another important topic for further investigation is to study detectability conditions of cyber attacks on the MAS systems that were studied in this paper.

\bibliographystyle{IEEEtran}
\bibliography{IEEEabrv,CCTARef}

\begin{thebibliography}{10}
\providecommand{\url}[1]{#1}
\csname url@samestyle\endcsname
\providecommand{\newblock}{\relax}
\providecommand{\bibinfo}[2]{#2}
\providecommand{\BIBentrySTDinterwordspacing}{\spaceskip=0pt\relax}
\providecommand{\BIBentryALTinterwordstretchfactor}{4}
\providecommand{\BIBentryALTinterwordspacing}{\spaceskip=\fontdimen2\font plus
\BIBentryALTinterwordstretchfactor\fontdimen3\font minus
  \fontdimen4\font\relax}
\providecommand{\BIBforeignlanguage}[2]{{%
\expandafter\ifx\csname l@#1\endcsname\relax
\typeout{** WARNING: IEEEtran.bst: No hyphenation pattern has been}%
\typeout{** loaded for the language `#1'. Using the pattern for}%
\typeout{** the default language instead.}%
\else
\language=\csname l@#1\endcsname
\fi
#2}}
\providecommand{\BIBdecl}{\relax}
\BIBdecl

\bibitem{1333204}
R.~{Olfati-Saber} and R.~M. {Murray}, ``Consensus problems in networks of
  agents with switching topology and time-delays,'' \emph{IEEE Transactions on
  Automatic Control}, vol.~49, no.~9, pp. 1520--1533, Sep. 2004.

\bibitem{4804653}
N.~{Meskin} and K.~{Khorasani}, ``Actuator fault detection and isolation for a
  network of unmanned vehicles,'' \emph{IEEE Transactions on Automatic
  Control}, vol.~54, no.~4, pp. 835--840, April 2009.

\bibitem{DAVOODI2016185}
M.~Davoodi, N.~Meskin, and K.~Khorasani, ``Simultaneous fault detection and
  consensus control design for a network of multi-agent systems,''
  \emph{Automatica}, vol.~66, pp. 185 -- 194, 2016.

\bibitem{dfdfi}
I.~Shames, A.~M. Teixeira, H.~Sandberg, and K.~H. Johansson, ``Distributed
  fault detection for interconnected second-order systems,'' \emph{Automatica},
  vol.~47, no.~12, pp. 2757 -- 2764, 2011.

\bibitem{semsar2009multi}
E.~Semsar-Kazerooni and K.~Khorasani, ``Multi-agent team cooperation: A game
  theory approach,'' \emph{Automatica}, vol.~45, no.~10, pp. 2205--2213, 2009.

\bibitem{xu2013consensus}
J.~Xu, L.~Xie, T.~Li, and K.~Y. Lum, ``Consensus of multi-agent systems with
  general linear dynamics via dynamic output feedback control,'' \emph{IET
  Control Theory \& Applications}, vol.~7, no.~1, pp. 108--115, 2013.

\bibitem{taheri2017adaptive}
M.~Taheri, F.~Sheikholeslam, M.~Najafi, and M.~Zekri, ``Adaptive fuzzy wavelet
  network control of second order multi-agent systems with unknown nonlinear
  dynamics,'' \emph{ISA Transactions}, vol.~69, pp. 89--101, 2017.

\bibitem{li2009consensus}
Z.~Li, Z.~Duan, G.~Chen, and L.~Huang, ``Consensus of multiagent systems and
  synchronization of complex networks: A unified viewpoint,'' \emph{IEEE
  Transactions on Circuits and Systems I: Regular Papers}, vol.~57, no.~1, pp.
  213--224, 2009.

\bibitem{adaii}
F.~Pasqualetti, F.~Dörfler, and F.~Bullo, ``Attack detection and
  identification in cyber-physical systems,'' \emph{IEEE Transactions on
  Automatic Control}, vol.~58, no.~11, pp. 2715--2729, Nov 2013.

\bibitem{ascf}
A.~Teixeira, I.~Shames, H.~Sandberg, and K.~H. Johansson, ``A secure control
  framework for resource-limited adversaries,'' \emph{Automatica}, vol.~51, no.
  Supplement C, pp. 135 -- 148, 2015.

\bibitem{barboni2019distributed}
A.~Barboni, H.~Rezaee, F.~Boem, and T.~Parisini, ``Distributed detection of
  covert attacks for interconnected systems,'' in \emph{2019 18th European
  Control Conference (ECC)}.\hskip 1em plus 0.5em minus 0.4em\relax IEEE, 2019,
  pp. 2240--2245.

\bibitem{docaazdaicps}
A.~Hoehn and P.~Zhang, ``Detection of covert attacks and zero dynamics attacks
  in cyber-physical systems,'' in \emph{2016 American Control Conference
  (ACC)}.\hskip 1em plus 0.5em minus 0.4em\relax IEEE, 2016, pp. 302--307.

\bibitem{siolcps}
A.~{Baniamerian} and K.~{Khorasani}, ``Security index of linear cyber-physical
  systems: A geometric perspective,'' in \emph{2019 6th International
  Conference on Control, Decision and Information Technologies (CoDIT)}, April
  2019, pp. 391--396.

\bibitem{feng2016distributed}
Z.~Feng, G.~Hu, and G.~Wen, ``Distributed consensus tracking for multi-agent
  systems under two types of attacks,'' \emph{International Journal of Robust
  and Nonlinear Control}, vol.~26, no.~5, pp. 896--918, 2016.

\bibitem{he2018secure}
W.~He, X.~Gao, W.~Zhong, and F.~Qian, ``Secure impulsive synchronization
  control of multi-agent systems under deception attacks,'' \emph{Information
  Sciences}, vol. 459, pp. 354--368, 2018.

\bibitem{ma2016consensus}
L.~Ma, Z.~Wang, and Y.~Yuan, ``Consensus control for nonlinear multi-agent
  systems subject to deception attacks,'' in \emph{2016 22nd International
  Conference on Automation and Computing (ICAC)}.\hskip 1em plus 0.5em minus
  0.4em\relax IEEE, 2016, pp. 21--26.

\bibitem{mustafa2019attack}
A.~Mustafa and H.~Modares, ``Attack analysis for discrete-time distributed
  multi-agent systems,'' in \emph{2019 57th Annual Allerton Conference on
  Communication, Control, and Computing (Allerton)}.\hskip 1em plus 0.5em minus
  0.4em\relax IEEE, 2019, pp. 230--237.

\bibitem{guan2017controllability}
Y.~Guan, Z.~Ji, L.~Zhang, and L.~Wang, ``Controllability of multi-agent systems
  under directed topology,'' \emph{International Journal of Robust and
  Nonlinear Control}, vol.~27, no.~18, pp. 4333--4347, 2017.

\bibitem{ji2006leader}
M.~Ji, A.~Muhammad, and M.~Egerstedt, ``Leader-based multi-agent coordination:
  Controllability and optimal control,'' in \emph{2006 American Control
  Conference}.\hskip 1em plus 0.5em minus 0.4em\relax IEEE, 2006, pp. 6--pp.

\bibitem{lou2012controllability}
Y.~Lou and Y.~Hong, ``Controllability analysis of multi-agent systems with
  directed and weighted interconnection,'' \emph{International Journal of
  Control}, vol.~85, no.~10, pp. 1486--1496, 2012.

\bibitem{rahmani2009controllability}
A.~Rahmani, M.~Ji, M.~Mesbahi, and M.~Egerstedt, ``Controllability of
  multi-agent systems from a graph-theoretic perspective,'' \emph{SIAM Journal
  on Control and Optimization}, vol.~48, no.~1, pp. 162--186, 2009.

\bibitem{zhang2013upper}
S.~Zhang, M.~Cao, and M.~K. Camlibel, ``Upper and lower bounds for controllable
  subspaces of networks of diffusively coupled agents,'' \emph{IEEE
  Transactions on Automatic control}, vol.~59, no.~3, pp. 745--750, 2013.

\bibitem{horn1991topics}
R.~A. Horn and C.~R. Johnson, ``Topics in matrix analysis cambridge university
  press,'' \emph{Cambridge, UK}, 1991.

\bibitem{man}
Y.~Desmedt, ``Man-in-the-middle attack,'' \emph{Encyclopedia of cryptography
  and security}, pp. 759--759, 2011.

\bibitem{communication}
D.~Juneja, A.~Singh, and A.~Jagga, ``Kqml based communication protocol for
  multi agent systems,'' \emph{International Journal of Emerging Technologies
  in Computational and Applied Sciences (IJETCAS)}, vol.~12, pp. pp. 158--162,
  05 2015.

\bibitem{pirani2014spectral}
M.~Pirani and S.~Sundaram, ``Spectral properties of the grounded laplacian
  matrix with applications to consensus in the presence of stubborn agents,''
  in \emph{2014 American Control Conference}.\hskip 1em plus 0.5em minus
  0.4em\relax IEEE, 2014, pp. 2160--2165.

\bibitem{brogan1991modern}
W.~L. Brogan, \emph{Modern control theory}.\hskip 1em plus 0.5em minus
  0.4em\relax Pearson education india, 1991.

\bibitem{willems1997introduction}
J.~C. Willems and J.~W. Polderman, \emph{Introduction to mathematical systems
  theory: a behavioral approach}.\hskip 1em plus 0.5em minus 0.4em\relax
  Springer Science \& Business Media, 1997, vol.~26.

\bibitem{actuatorplacement}
T.~H. Summers and J.~Lygeros, ``Optimal sensor and actuator placement in
  complex dynamical networks,'' \emph{IFAC Proceedings Volumes}, vol.~47,
  no.~3, pp. 3784--3789, 2014.

\bibitem{slotine}
Y.-Y. Liu, J.-J. Slotine, and A.-L. Barab{\'a}si, ``Control centrality and
  hierarchical structure in complex networks,'' \emph{Plos one}, vol.~7, no.~9,
  p. e44459, 2012.

\bibitem{leader1}
A.~Olshevsky, ``Minimal controllability problems,'' \emph{IEEE Transactions on
  Control of Network Systems}, vol.~1, no.~3, pp. 249--258, 2014.

\bibitem{leader2}
K.~Fitch and N.~E. Leonard, ``Optimal leader selection for controllability and
  robustness in multi-agent networks,'' in \emph{2016 European Control
  Conference (ECC)}.\hskip 1em plus 0.5em minus 0.4em\relax IEEE, 2016, pp.
  1550--1555.

\end{thebibliography}

\end{document}